\documentclass[letterpaper,10pt]{llncs}

\usepackage[greek,english]{babel}
\usepackage{amssymb}
\usepackage{soul} 
\usepackage{xcolor} 

\usepackage{subfig}
\usepackage{caption}
\usepackage{subfloat}

\usepackage[pdftex]{graphicx}
\usepackage{epstopdf}

\usepackage{algorithm}
\usepackage{algpseudocode}

\newenvironment{localscope}{}{}
\newcommand{\AComment}[2][0pt]{}
\newcommand{\cw}[1]{}

\usepackage{color}

\newcounter{todocnt}

\newcommand{\hla}[1]{\hl{#1}}
\newcommand{\hlb}[1]{\hl{#1}}

\renewcommand{\hla}[1]{#1}
\renewcommand{\hlb}[1]{#1}

\begin{document}
\title{On Money as a Means of Coordination between Network Packets}
\author{Pavlos S. Efraimidis and Remous-Aris Koutsiamanis}
\institute{Dept. of Electrical and Computer Engineering \\
Democritus Univ. of Thrace, 67100 Xanthi, Greece \\
\{pefraimi, akoutsia\}@ee.duth.gr}

\maketitle

\begin{abstract}
In this work, we apply a common economic tool, namely money,
to coordinate network packets. In particular, we present a network economy, called PacketEconomy,
where each flow is modeled as a population of rational network packets, and these packets can self-regulate their access to network resources by mutually trading their positions in router queues. 
\hla{Every packet of the economy has its price, and this price determines if and when the packet will agree to buy or sell a better position.} We consider a corresponding Markov model of trade and show that there are Nash equilibria (NE) where queue positions and money are exchanged directly between
the network packets. This simple approach, interestingly, delivers improvements even when fiat money is used. We present theoretical arguments and experimental results to support our claims. 
\end{abstract}

\section{Introduction}
In their seminal work, Kiyotaki and Wright~\cite{KW89} examine the emergence of money as
a medium of exchange in barter economies. Subsequently, Gintis~\cite{Gi97,Gi00} generalizes
the Kiyotaki-Wright model by combining Markov chain theory and game theory.
Inspired by the above works, we propose the PacketEconomy where money is used 
as a coordination mechanism for network packets and prove that there are Nash equilibria where trades are 
performed to the benefit of all the flows.
In the PacketEconomy, specialization - the reason for the emergence of money as per Adam Smith (\cite[Chapter 4]{Sm1776}, cited in~\cite{KW89}) - 
originates from the diverse QoS requirements of
network flows. More precisely, various types of flows differ in their
tolerance for packet delays.

It is known that a large number of independent flows is constantly
competing on the Internet for network resources.
Without any central authority to regulate its
operation, the available network resources are allocated by
independent routers to the flows in a decentralized manner.
In this environment an
Internet flow may submit at any time an arbitrary amount of work to the network.
If the packets of the flow are successfully delivered, the flow may continue with the same
sending rate or even gradually increase it.
In case of congestion at one or more nodes in the flow's path, the
flow experiences delays and packet losses.
Then, the flow is expected to reduce its packet rate with an appropriate flow control algorithm,
like the AIMD-based algorithms for TCP-flows\footnote{AIMD stands for Additive Increase Multiplicative Decrease
and TCP for the Transmission Control Protocol of the Internet protocol suite.}.
This way, each flow independently decides on its submission rate.
The apparent lack of coordination between the independent
flows leads the Internet to an ``anarchic'' way of operation
and gives rise to issues and problems that can be addressed with
concepts and tools from algorithmic game theory.

Two representative works on applying game theory to network problems
are~\cite{KP99,Pa01}.
Certain game-theoretic approaches to congestion problems
of the Internet, and especially the TCP/IP protocol suite,
are discussed in~\cite{Sh95,ASKSP02,GJS04,ETM10}.
A combinatorial perspective on Internet congestion problems is
given in~\cite{KKPS00}. The focus of the above works and the present 
paper is on sharing the network resources between selfish flows.
In this work, however, we propose an economy where packets belonging to selfish flows may
interact directly with each other.

The use of economic tools like pricing, tolls and taxes as a means to regulate
the operation of networks and/or to support quality of service (QoS)
functionalities in the presence of selfish flows is, for example, discussed
in~\cite{Od99,GK95,CDR03STOC,CDR03EC,MV93,MB97}. In particular, the
Paris Metro Pricing approach - using pricing to manage traffic in the Paris Metro -
is adapted to computer networks in~\cite{Od99}. A smart market
for buying priority in congested routers is presented in~\cite{MV93}.
In~\cite{CDR03STOC,CDR03EC} taxes are used to influence the behavior of
selfish flows in a different network model. An important issue identified
in~\cite{CDR03EC} is that taxes may cause disutility to network users unless
the collected taxes can be feasibly returned to the users.
In our economic model this issue is naturally solved; trades take place
between the flows, so the money is always in the possession of the flows.

In this work, we apply a common economic tool, namely money,
to coordinate network packets. This is in contrast to much of the existing literature,
which aims to impose charges on Internet traffic, and to our knowledge,
this is the first work to propose economic exchanges directly between packets.
In particular, we present a network economy, called PacketEconomy, where
ordinary network packets can trade their positions in router queues.
The role of money in this approach is to facilitate the trades between the network packets.
Queue positions and money are exchanged directly between
the packets while the routers simply carry out the trades.
We show that, in this economy, packets can self-regulate their
access to network resources and obtain better services at equilibrium points.

\vspace{0.2cm}
\noindent
{\bf Contribution.} The main contributions of this work are:
\begin{itemize}
\renewcommand{\labelitemi}{$\bullet$}
\item A new game-theoretic model representing network packets as populations of rational agents.
In this model a network flow is represented as a population of in-flight packets that can make 
bilateral trades with other packets.
\item Application of bilateral trades and virtual money at a microeconomic level to  
support better coordination of rational network packets.
\item A plausible model for a very important and challenging problem like 
TCP/IP packet scheduling. The model may also apply to other decentralized 
coordination problems.
\item Application of an interesting combination of ergodic Markov chains and strategic
games within a new context.
\item Theoretical evidence and experimental results showing that the model has 
desirable Nash equilibria.
\end{itemize}

\vspace{0.2cm}
\noindent
{\bf Outline.} The rest of this work is organized as follows. 
Preliminaries are given in Section~\ref{sec:prelim}.
In Section~\ref{sec:EconomyForPackets} we describe the PacketEconomy.
Afterwards, we analyze two representative scenarios of the PacketEconomy
in Section~\ref{sec:scenarios} and discuss the underlying packet scheduling 
problem in Section~\ref{sec:sched}. 
The effect of trades in the PacketEconomy is examined in Section~\ref{sec:effect}. 
We then present an experimental evaluation of the PacketEconomy in Section~\ref{sec:exp} 
and make concluding remarks in Section~\ref{sec:conclusion}. 

\section{Preliminaries}
\label{sec:prelim}

In this section, we present a brief description of concepts that 
are required to understand the motivation and the results of this work.
More precisely, we discuss some introductory concepts of game theory, 
TCP networking, and modern Internet routers.
Readers familiar with some or all of these topics may skip the 
corresponding material.

\subsection{Game Theory Terminology}

\hla{
For convenience we define some basic concepts of game theory and optimization
as they are used in this work. More details on these introductory concepts can be found
in any related textbook.
A game is a mathematical model of the interaction among rational, mutually aware players.
In this work, selfish, strategic and rational are used interchangeably to denote players
whose objective is to maximize their own payoff.
The payoff of each player is determined by the outcome of the game, which in turn depends
on the decisions (strategies) of all players.
A strategy defines a set of moves or actions a player will follow in a given game.}

\hla{A mixed strategy is a randomized strategy that assigns a probability to each pure
strategy. The support of a mixed strategy is the set of actions to which it assigns a
strictly positive probability. 
A strategy profile is a set of strategies that includes one and only one strategy for every player.
Clearly, a strategy profile fully specifies a single execution of a game.
A Nash equilibrium is a strategy profile were no player has an incentive to unilaterally 
deviate from tis strategy.
}

\hla{
We also refer to the concept of a weak Pareto improvement, which (in this context) is
any change to the current strategy profile that makes every player at least as well off 
and at least one player strictly better off.
}

\subsection{Some TCP/IP Terminology}
\label{sec:tcp}
The Transmission Control Protocol (TCP) belongs to the Internet protocol suite and is the main reliable connection-oriented data transmission
protocol of the Internet. TCP flows transmit their data by sending a series of packets. Assume a TCP flow that is ready to send
a large volume of data as a sequence of 220 packets, each of size 1024 bytes. 
In order to send the data in a controlled manner, a first parameter $w$ is used,  called the size of the \emph{congestion window}.

The TCP protocol dictates that the flow starts by submitting $w$ packets to
the network and then either waits until either a packet’s arrival is confirmed, 
normally by receiving a matching acknowledgement packet (ACK) within a certain time-frame,
or the time-frame passes, whereby the packet is considered lost.
As soon as the number of the in-flight packets of the flow is less than $w$, the flow submits
new packet(s); the result is that, at any moment in time, the flow can have at most $w$ packets in flight. Thus, the size $w$ of the
congestion window has a strong impact on the transmission rate of a flow~\cite{Ja88}. Consequently, the selection of an appropriate value for $w$
is a very critical task for each flow, and this is where the AIMD (Additive Increase Multiplicative Decrease) scheme is useful. 

The AIMD algorithm
is the most popular procedure for a TCP flow to constantly adapt its window size to the changing network conditions.
The basic principle of AIMD is that, for each successful packet delivery the flow increases its congestion window size additively
by an amount proportional to a parameter $\alpha > 0$ (usually $\alpha = 1$) and for each lost packet, the flow decreases its congestion
window multiplicatively by a parameter $\beta < 1$ (usually $\beta = 1/2$). The values of the $\alpha$ and $\beta$ parameters have a decisive
role on the behavior of the AIMD flow. A large value of $\alpha$ and/or $\beta$ makes the flow more aggressive, whereas a small value
makes it more temperate. More details can be found in common computer networks books like~\cite{Peterson:2007:Networks,Stevens:1994:Vol1}.


\subsection{Network Routers}
\label{sec:hw}
Hardware-based routers, such as those commonly produced by Cisco, fall into two large categories based on their maximum throughput: High-end routers and medium/low-end routers.
High-end routers are typically employed in backbone networks and thus need to support extremely high throughput. To achieve this, they employ fixed-function dedicated and highly parallel hardware computation units (NPUs) as well as specialized high-speed memory (TCAM).
However, this comes at the cost of flexibility and customizability, as the algorithms which can be used by the router while maintaining its high-speed processing as predetermined and implemented into hardware. Some parameters may be configurable but only to the extent predetermined by the manufacturer. Often, for the target application these limitations may not be a problem, since backbone routers often do not have enough context in order to make flow-dependant routing choices. For example, one limitation which affects our system as well, is that it is impossible to perform packet re-ordering within the queue (the queue is strictly FIFO).
If higher flexibility is desired, it \emph{is} possible in many cases to use custom algorithms within these routers, however this is done at the expense of bypassing a part of the hardware-based pipeline through a software-based one. The immediate effect is that throughput drops significantly.

While these trade-offs have to do with high-end backbone routers, lower-cost middle- and low-end routers, which do not need to provide the same throughput since they are typically used near the leafs of the network, largely do away with the specialized and costly hardware implementation and use a software pipeline. As a result, it is much easier to implement custom algorithms on this class of routers.

\section{An Economy for Packets}
\label{sec:EconomyForPackets}
The PacketEconomy is comprised of a network model with selfish flows,
a queue that supports packet trades, a currency and a specific economic goal.
The solution concept is the Nash equilibrium (NE), i.e., 
a profile of the game in which no player has anything to gain by changing only his/her own strategy unilaterally.

\begin{figure}[h!]
\begin{minipage}[b]{0.38\linewidth}
\centering
\includegraphics[width=\textwidth]{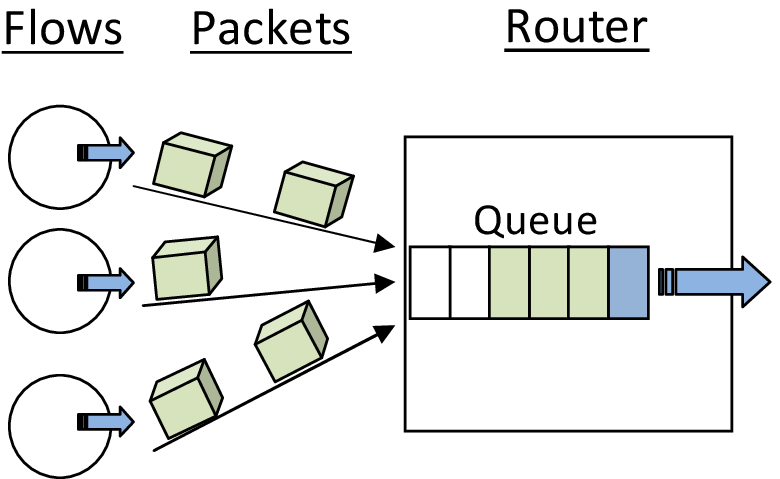}
\caption{The network model with the flows and their packets, the router, and the router queue.\\}
\label{fig:NetworkModel}
\end{minipage}
\hspace{0.2cm}
\begin{minipage}[b]{0.55\linewidth}
\centering
\includegraphics[width=\textwidth]{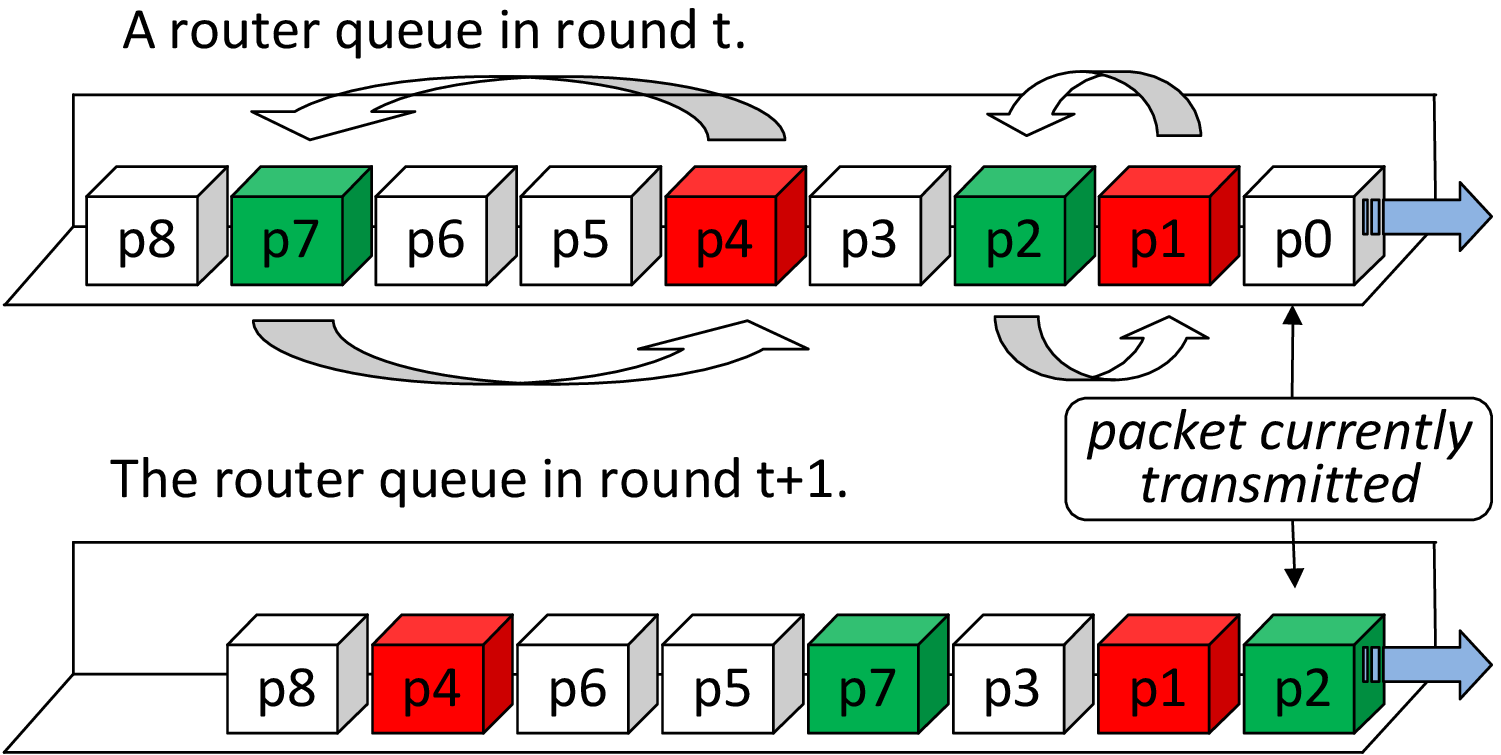}
\caption{The state of a router queue in two successive rounds, t and t+1. Two trades are performed in round t; 
one between the packet pair (p1,p2) and one between the pair (p4,p7).}
\label{fig:SimpleTrade}
\end{minipage}
\end{figure}

\textbf{The Network Model.}
We assume a one-hop network with a router R and a set of N flows, as shown in Figure~\ref{fig:NetworkModel}.
This setting is equivalent to the common dumbbell topology used for the analysis of many network scenarios,
including the seminal paper of Chiu and Jain~\cite{CJ89} on the AIMD algorithm.
The router R has a FIFO DropTail queue with a maximum capacity of Q packets and operates in rounds.
In each round the first packet (the packet at position 0 of the queue) is served.
At the end of the round the first packet reaches its destination.
Packets that arrive at the router are added to the end of the queue.

\textbf{Packet Trades.}
At the beginning of each round all packets in the queue are shifted one position ahead.
A packet that enters the queue in this round, occupies the first free (after the shift) 
position at the end of the queue. After the move and for the rest of the duration of the round,
the packet that has reached position zero is served.
At the same time, the other packets
in the router queue are simply waiting. These idle packets can engage in trades.
During each router round a fixed number $b$ of trading periods take place.
In each trading period the idle packets are matched randomly in pairs
with a predefined pairing scheme.
Each packet pair can perform a trade, as shown in Figure~\ref{fig:SimpleTrade},
provided
the negotiation performed between them leads to an agreement.
The way the trades take place at a microeconomic level
between paired packets resembles the models of~\cite{Gi97,KW89}
where agents meet in random pairs and can make trades.
In the theoretical analysis we assume a random-pairing scheme,
which corresponds to a well-mixed population, where packets 
are equally likely to be paired with any other packet. 

\textbf{Packet Delay.}
The packet delay $d_p$ of a packet $p$ that starts at position $k$
of the zero-based queue and does not make any trade is $k+1$ rounds (Figure~\ref{fig:delays}).
If, however, the packet engages in trades and buys a total of $r_b$ router rounds
and sells $r_s$ router rounds, then its packet delay $d_p$, including the time to be served,
becomes $d_p = k + 1 + r_s - r_b$ rounds. A packet may have an upper bound $d_{p,\max}$ on its
delay; for delays larger than $d_{p,\max}$ the value of the packet becomes zero and the packet 
will not voluntarily accept such delays (that is, it will not sell).

\textbf{Details.}
The router operates in rounds and can serve one packet in each one.
All packets are assumed to be of the same size and no queue overflows occur.
In generating the random packet pairs the use of predefined pairing reduces
the computational burden and avoids stable marriage problems.
We make the plausible assumption that flows with different QoS
preferences are competing for the network resources. We also make the
assumption that the preferences of each flow can be expressed with a utility function
for its packets. Thus, packets with different utility functions will, in general, co-exist in 
the router queue.

\textbf{Packet Values.}
For each packet there is a flow-specific decreasing value function $v_p(d)$
which given the delay $d$ of a packet reveals its value. 
The value function of each flow must be encoded
onto each packet and its computational requirements should be low in order not
to overload the router. A class of simple value functions are
$v_p(d) = \max \{ v_{\max} - c_p \cdot d, 0 \}$ where $c_p$ is the cost per unit
of delay (Figure~\ref{fig:values}).
Anytime during the packet's journey, its value can be estimated via the $v_p(d)$ function.
However, when the packet arrives at the destination, its value is finalized. 

\emph{In the PacketEconomy every packet has its compensatory price $p$. For prices lower than $p$,
the packet is ready to buy better queue positions and for higher prices higher than $p$ it is ready 
to sell its position, given that the extra delay will not cause it to exceed its maximum delay limit.
}

\begin{figure}[!h]
\subfloat[Packet delay terminology for p4]
{\label{fig:delays}\hspace{0.04\textwidth}\includegraphics[width=0.45\textwidth,height=0.25\textwidth]{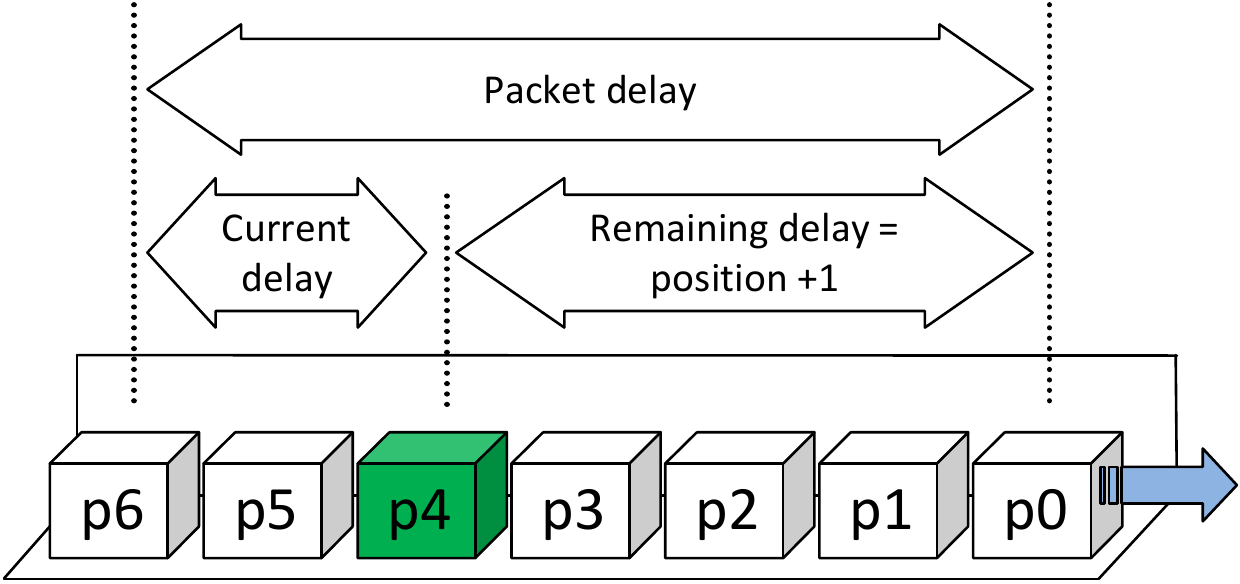}\hspace{0.01\textwidth}}
\hspace{0.02\textwidth}
\subfloat[Two simple packet value functions]
{\label{fig:values}\hspace{0.01\textwidth}\includegraphics[width=0.42\textwidth,height=0.25\textwidth]{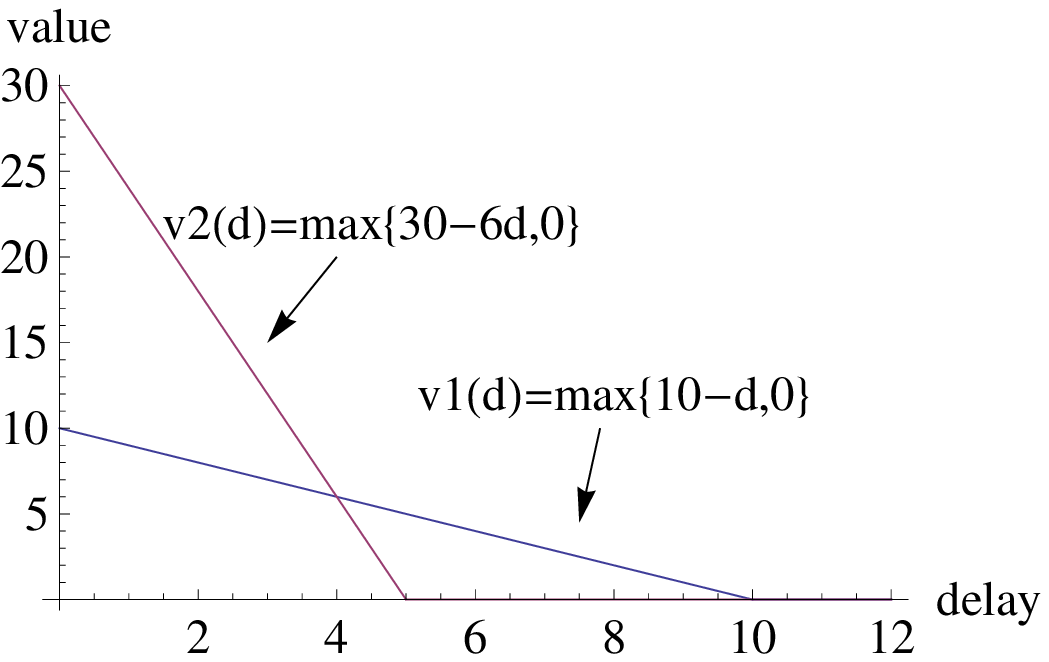}\hspace{0.04\textwidth}}
\caption{Delays and Packet Values.} \label{fig:DelayValues}
\end{figure}

\textbf{Inventories.}
Every time a packet is delivered in time, wealth is created for the flow that owns
the packet. Each packet $p$ has an inventory $I_p(t)$ containing two types of indivisible 
goods or resources;
the packet delay $d_p(t)$ and the money account $a_p(t)$.
Note that delay bears negative value, whereas money represents positive value.
We assume positive integer constants $s_a$, $s_b$ and $s_d$, such that 
$a_p(t) \in \{-s_a,\ldots,s_b\}$ and $d_p(t) \in \{0,\dots,s_d\}$. 
The inventory also contains the current position $pos_p(t)$ of the packet in the queue
(i.e., the remaining delay until it is served) if it is waiting in the queue.
When the packet reaches its destination, the contents of
the inventory of the packet are used to determine its utility (amount of wealth).
This utility is then reimbursed to the flow that owns the packet and a new packet
of the same flow enters the queue.
\hla{An inventory is called admissible, if the delay of the packet does not exceed its maximum 
delay. A packet would not agree to trade an admissible inventory state for a non-admissible one.
We assume that all packets start with an admissible inventory when they enter the queue.}

\textbf{Benefit and Utility.}
Since a packet has two types of resources that bear value, that is, the packet value and the budget of a packet,
we define the notion of the packet benefit as the sum of the value of a packet plus/minus the budget of the packet.
Then we use the benefit concept to define the utility function of the packet.
For fixed rate flows the benefit and the benefit rates
are equivalent and can either be used as the utility of the packet.
For Window-based flows the utility function is the benefit rate, i.e. the benefit per router round.

\textbf{Trades.}
The objective of each packet is to maximize its utility.
Thus, when two packets are paired in a trading period, their inventories and their
trading strategies are used to determine if they can agree on a mutually profitable trade,
in which one packet offers money and the other offers negative delay.
The obvious prerequisite for a trade to take place is that both packets prefer their
post-trade inventories to their corresponding pre-trade inventories.
For this to be possible there must be ``surplus value'' from a potential trade.
In this case both packets can benefit, i.e., increase their utility, if they come to
an agreement.

\subsection{Flow Types and the Cost of Delay}
The delay that a packet experiences has a direct impact on its value, which is a decreasing function of the delay. 
This impact depends on the type of the flow. 
Window-based flows employ a feedback-based mechanism, the congestion window, which determines the maximum number 
of packets that the flow may have in-flight. A brief description on these TCP/IP related concepts can be found in Section~\ref{sec:tcp}.
The consequence of using a congestion window is that there is an additional, secondary, effect of the packet delay on the flow's wealth. Every packet that is in-flight
occupies one of the available positions in the congestion window of a window-based flow. The more one
packet delays its arrival, the longer a following packet will have to wait to use the occupied window position. Therefore, the impact of packet delays for window-based flows is twofold; the decreased value of the delayed packet
and the reduced packet rate. 
On the other hand, for rate-based flows which submit packets with some given rate, there is no other consequence due to packet delays beyond the reduced value of 
the packets.

If $d = d_p(t)$ and $p = pos_p(t)$ are the delay and the position of a packet, 
then a trade that changes the delay to $d'$, also changes the 
value of the packet from $v_p(d)$ to $v_p(d')$. For rate-based flows, the difference
between these two values determines the compensatory price for the packet.

Assume a rate-based packet $p$ with balance $\alpha_1$ and 
delay $d_1 < d_{p,\max} - d_\epsilon$, for some $d_\epsilon$.
When a trade changes the delay from $d_1$ to $d_2=d_1+d_\epsilon$, then this 
also changes the value of the packet from $v(d_1)$ to $v(d_2)$. 
The difference between these two values determines the compensatory price 
$p$ for the packet.
\begin{equation}
\label{equ:priceRateBased}
p = v(d_1) - v(d_2) = v(d_1) - v(d_1 + d_\epsilon) = c_p d_\epsilon \; . 
\end{equation}
At this price, the utility of the packet remains unchanged after the trade.
A packet would agree to sell for a price $\rho_s > \rho$, or to buy for $\rho_b < \rho$.

For window-based flows, however, the price estimation needs more attention.
In this case the average benefit per round (benefit rate) is an appropriate measure 
for the utility of each packet of the flow. 
Assume a window-based packet $p$ with delay $d_1 < d_{p,\max} - d_\epsilon$, 
value $v_1 = v(d_1)$ and account balance $\alpha_1$.
Before the trade the utility (benefit rate) is $r_1 = (v_1+\alpha_1)/d_1$. 
If the packet agrees to
trade its position and to increase its delay by $d_\epsilon$, then
the new benefit rate is $r_2 = (v_2+\alpha_2)/d_2$. 
Let $\rho$ be the compensatory price for the trade. Then, by setting $r_1=r_2$ we obtain
\[
\frac{v_1+\alpha_1}{d_1} = \frac{v_2+\alpha_2}{d_2}
\Rightarrow \frac{V - c_p d_1 +\alpha_1}{d_1} = \frac{V - c_p (d_1+d_\epsilon)+(\alpha_1 + \rho)}{d_1+d_\epsilon} \Rightarrow \]
\begin{equation}
\label{equ:priceWindowBased}
 \rho = (V+\alpha_1)\frac{d_\epsilon}{d_1} \; .
\end{equation}

\noindent The above expression for the price ensures that the utility function
of the packet remains unchanged.
A packet would agree to sell its position, for a price $\rho_s > \rho$, or to buy a
position ($d_\epsilon < 0$) for $\rho_b < \rho$. 

\section{Equilibria with Monetary Trades}
\label{sec:scenarios}

\newcommand{\allKN}[1]{\forall #1 \in \mathbb{N}}
\newcommand{\allKX}[2]{\forall #1 \in \{1 \ldots #2 \}}
\newcommand{\allKXS}[3]{\forall #1 \in \{#2 \ldots #3 \}}

In this section, we illustrate the PacketEconomy approach in two representative scenarios
and discuss some auxiliary tools that are needed to support the operation of the PacketEconomy.

\subsection{Scenario 1}
\label{sec:ScenarioNoFiat}
This is a simple scenario that produces an interesting configuration.
It consists of a set of $N$ window-based flows $f_i,$ for $i \in \{1\dots N\}$,
each with a constant window size $w_i$, and $\sum_i w_i = q$.
When a packet is served by the router it is immediately replaced by an identical packet
submitted by the same flow. This is a simplifying but plausible assumption. In reality,
when a flow packet arrives at its destination, a small size acknowledgment packet (ACK) 
is submitted by the receiver. When the sending flow receives the ACK it submits a new
identical packet that immediately enters the queue.
We assume $b=1$ trading period 
per router round but in general there can be any constant number of trading periods per router 
round.

\textbf{Failure states.}
For each packet there is a small probability $p_f$ for an extra delay of $d_f$ rounds, where $d_f$ is a discrete random variable in $\{1,2,\dots,q-1\}$.
These delays correspond to potential packet failures of real flows, and occur
between the service of a packet and the submission of its replacement.
The presence 
of these failures will be useful to show the ergodicity of the Markov process of the economy.
By convention, the delay $d_f$ is added to the delay of the packet that has just been served. If more than one packets enter the queue at the same time (synchronized due to delays), 
their order in the queue is decided upon uniformly at random. 
A packet that does not participate in any trade and does not suffer delay due to failure will
experience a total delay of $q$ rounds.

{\bf Packet states and strategies.}
The state $\tau_p(t)$ of a packet $p$ in round $t$ is a pair
$\tau_p(t)=(I_p(t), rel_p(t))$, where $I_p(t)$ is the inventory of the packet and $rel_p(t)$, which is meaningful only in failure states, is the remaining number of failure rounds for the packet.
The state of all packets of the economy in round $t$ determines the state of the whole economy $\tau(t) = \prod_{p=0}^{q-1} \tau_p(t)$.
From a packet's point of view, a trade is simply an exchange of its inventory state (budget, delay and position)
with a new one. Consequently, a pure strategy of a packet is a complete ordering of the possible states
of its inventory. 
In each round the waiting packets move by default one position ahead and, thus, enter 
a new inventory state. 
We assume that the packet ignores the impact of its state and strategy on the state of the packet population.
In every trading period the packet simply aims at myopically improving its state.

{\bf A trade.}
Assume that two packets $p1$ and $p2$ trade with $pos_{p1}(t) < pos_{p2}(t)$, that is, $p1$ is closer to the serving end of the queue than $p2$. Then, in the next round $t' = t+1$, it holds that $pos_{p1}(t') = pos_{p2}(t) - 1$ and $pos_{p2}(t') = pos_{p1}(t) - 1$. If $\rho$ is the agreed upon price of the trade, then also $a_{p1}(t') = a_{p1}(t) + \rho$ and $a_{p2}(t') = a_{p2}(t) - \rho$. A prerequisite for the trade to take place is that $\rho_s \leq \rho_b$, where
$\rho_s$ is the asking price of the seller and $\rho_b$ the offered price by the buyer.
To that end each packet has to calculate the utility loss/gain it will incur if it trades its position,
taking into account the impact on the packet value $v_p$ and on the rate of the packet (for window-based flows).

\begin{definition}
Let $\tau(t)$ be the state of the economy in round $t$.
\end{definition}

\begin{lemma}
\label{lem:Markov}
$\tau(t)$ is an ergodic Markov chain.
\end{lemma}
\begin{proof}
Assume $b=1$ trading period per round.
In each round the economy moves to a new state with transition probabilities that depend only
on the current state and the strategies of the packets. Let $\sigma_p$ be a pure strategy of
each packet $p$ of a flow and $\sigma$ be a pure strategy profile of the whole economy. 
Then, there is a corresponding transition probability matrix $P^{\sigma}$ for the economy.
Let $\sigma_m$ be a mixed strategy profile of the whole economy. Then the corresponding
transition probability $P^{\sigma_m}$ of the economy for $\sigma_m$ 
is an appropriate convex combination of the transition matrices of the supporting pure strategies. 
In case of multiple trading periods per round $(b > 1)$, the economy makes $b$ state transitions 
per round.

The number of potential states for a packet is finite and, consequently, the number of states 
for the whole economy is also finite. 
\begin{definition}
\hla{A zero state $\tau_0$ is a state of the economy in which all packets have zero budget and each packet $p$ has delay $d_p(t_0) = pos_p(t_0) + 1$, where $t_0$ is the current round of the router. }
\end{definition}
Assume that in round $t$ the packet at position $0$ fails for $q-1$ rounds, in round $t+1$ the next packet at position $0$ fails for $q-2$ rounds etc. 
Then after $q$ rounds all new packets will simultaneously enter the queue. Each packet will have zero budget
and by definition their ordering will be random. This also means that for each packet $p$,  $d_p(t) = pos_p(t) + 1$. Thus, in round $t+q$ the economy will be in a zero state.
The probability for this to happen is strictly positive and thus each zero state $\tau_0$ is recurrent. Since the number of states of the economy is finite, the states that are attainable from zero states like $\tau_0$ form a finite, aperiodic, irreducible set of states. 
From the Fundamental Theorem of Markov Chains (see for example~\cite{Motwani:1995:RA} or~\cite{Mitzenmacher:2005:PCR}) we know that any finite, irreducible, and aperiodic Markov chain is ergodic. This completes the proof of Lemma~\ref{lem:Markov}.
\end{proof}

\begin{lemma}
\label{lem:sce1pi}
For each pure strategy profile $\sigma$ of the economy, there is a unique stationary distribution $\pi_\sigma$ of the economy.
\end{lemma}
\begin{proof}
For each pure strategy profile $\sigma$, the Markov chain of the economy has a finite number of states, is aperiodic and ergodic. Thus, it must have a unique stationary distribution $\pi_\sigma$
(see for example~\cite[Theorem 7.7]{Mitzenmacher:2005:PCR}).
\end{proof}

\noindent
By following this line of reasoning, an interesting argument which can be applied is that
given the stationary distribution of the economy, each trading period becomes a finite state game. 

\begin{lemma}
\label{lem:FiniteGame}
For every idle packet, each trading period of the economy corresponds to a finite strategic game.
\end{lemma}

\begin{proof}
Let $\sigma_m$ be a mixed strategy of the whole economy and $P^{\sigma_m}$ the corresponding transition matrix of the Markov chain of the economy. 
\hla{Note that $P^{\sigma_m}$ is a convex combination of the transition
matrices $P^{\sigma}$ that correspond to the pure strategies $\sigma$ in the support of $\sigma_m$.}
Moreover, let $\pi_{\sigma_m}$ be the stationary distribution of the Markov chain for
transition matrix $P^{\sigma_m}$. We assume that the utility of each player (packet) 
for the profile $\sigma_m$ is the expected value of the player in the stationary 
distribution $\pi_{\sigma_m}$. In this way, we obtain for each trading period a finite 
game where every packet of the queue is a player. The strategy of the packet is its 
trading strategy. 
\end{proof}

\noindent
This leads us to the following theorem, which holds under plausible assumptions.

\begin{theorem}
\label{the:scenario1}
A NE of the economy exists where packets perform trades.
\end{theorem}

\begin{proof}
Since each trade is a finite game, the classic theorem of Nash~\cite{Nash:1950:Equilibrium,Nash:1951:NCG} assures that there is at least one mixed Nash equilibrium. However, the state of the economy where no packet participates in trades is a trivial
NE where no trades take place. We have to show that there at least one more NE.
A nice property of the current proof technique (due to Gintis~\cite{Gi97}) is that we can impose conditions on the 
equilibrium point. We can assume a restricted version of the economy, where each packet has a non-empty pure trading strategy set. In a sense each packet is enforced to accept at least some types of profitable trades every time it is possible. 

In the restricted economy each round is again a finite game and, consequently, it has a mixed NE. This time the NE has trades assuming that packets with different utility functions exist in the population. Assume now a NE state of the restricted game in the original unrestricted economy. 
\hla{It can be shown that,} assuming appropriate utility functions for the packets, 
if we relax the forced-trade restriction, then no packet has an incentive to 
unilaterally change its strategy.
That is, there exists a NE with trades for the original PacketEconomy.
\end{proof}

\subsection{Scenario 2}
\label{sec:ScenarioFiat}
We examine now a scenario where fiat money, that is, money without any intrinsic value~\cite{Mankiw:2008:Economics}, is used. Fiat money is by definition an object that 
is inherently worthless. This means that in the game-theoretic context of the PacketEconomy,
fiat money will not appear in any utility function and will not be axiomatically redeemable as anything else. We define and analyze this scenario to further emphasize the potential of money as a coordination tool.

We assume $N$ players, each with two flows, a business flow and an economy flow.
Let $m$ be a fixed quantity of fiat money available in the economy. Every packet has 
a value function and the objective of a player is to maximize the total utility
per round from its two flows. As implied by its definition, fiat money is not part of the utility
function.
Each player defines the trading strategies for its flows. Economy flows accept fiat
money to trade their queue position, while business flows spend fiat money to reduce their
delay. The money collected by the economy flow is used by the player to finance the  
business flow.

The two flows of each player act as a team with the common objective to maximize the total
utility obtained by the two flows together. For simplicity we assume window-based flows with 
fixed window size 1. This leads to a team game (a concept first studied in~\cite{Palfrey:1983:Team}) 
where teams of two player-packets each, compete for maximum utility.
The packets of each team collaborate whereas the teams participate in a non-cooperative game. 

An important difficulty of such a team game for the PacketEconomy is the lack of 
common information between the packets of the same team. A business
packet that is ready to spent money cannot be informed about the exact amount 
of money its economy packet team member currently has. 
Analyzing the scenario as a team game would be interesting
but is out of scope for the present work. 
Instead, we will simplify and formulate this setting as a normal strategic form game. 

{\bf The adapted economy.} We define an adapted version of the economy where 
each flow is an independent player and for each of the original flow pairs there
is a common deposit located at the common source of the two flows. Every time a packet is served
and before its successor is submitted, it interacts with the deposit of the pair and liquidates or buys fiat money. There is a bound $s_d$ on the amount of fiat money that can be stored in the deposit.
An economy packet liquidates and a business flow buys as much money as possible. The value of fiat money for the transactions of the team members with 
the deposit is fixed at $(c_e + c_b)/2$, where $c_e$ and $c_b$ are the costs per unit
of delay for the two packets. This assumption restricts the solution space of the game
but admits us to analyze it as a non-cooperative game. Note that if it were analyzed as a team game, 
as long as the price ratio $\rho_e / \rho_b$ would satisfy $\rho_e / \rho_b < c_e/c_b$, the flow pair would gain from each trade. 

Packet failures that increase the packet delay can occur just like in Scenario 1. 
However, in this scenario when both packets of a team are in a failed state in the same round, 
then we assume that a team failure has occurred and the team is reset. The deposit is set to zero.
Since fiat money must be preserved in the economy (else fiat money would disappear from
the economy after a finite number of rounds), the fiat money of the team is handed
over to the router who distributes it randomly in the next round to the packets in the
router queue (taking care not to violate maximum budget bounds of any packet).

Under plausible assumptions and similarly to Theorem~\ref{the:scenario1} we get:
\begin{theorem}
\label{the:scenario2}
A NE of the adapted economy exists where packets trade fiat money.
\end{theorem}
\begin{proof}[Sketch]
\begin{itemize}
\item The adapted economy can be modeled as a finite state Markov chain. The proof is similar to the proof of Lemma~\ref{lem:Markov}. The main difference is that the state of the economy comprises also the state of the common deposits of the flow pairs.
\item For each mixed strategy of the economy, there is a corresponding stationary distribution of the Markov chain of the economy. The proof is similar to the proof of Lemma~\ref{lem:sce1pi}. 
\item At the stationary distribution, each trading period corresponds to a finite strategic game.
The proof is similar but more involved than the proof of Lemma~\ref{lem:FiniteGame}.
\hla{First, we have to show that at the stationary state, each trading period corresponds to the 
normal form game $G$ defined above. Then, we have to show that a NE of the normal form game 
corresponds to a NE of a restriction of the original game, where the value of fiat money
is fixed within each flow pair.}
\item There is an equilibrium of the adapted economy, where fiat money circulates.
\end{itemize}
\end{proof}

Theorem~\ref{the:scenario2} \hla{shows that there is a NE where trades take place and fiat money 
circulates in the adapted economy. Clearly, the equilibrium state of the adapted economy is 
a valid state of the original economy which is based on the team-game with the flow pairs. 
However, it is open to find out if the NE of the adapted economy is a NE of the original economy too.
Moreover, it is possible that the free exchange rate for fiat money in the original economy may lead to more efficient (better wealth rate) equilibrium points and/or NE than the adapted economy.}

\subsection{Auxiliary Tools}

\noindent
{\bf Money as a Coordination Tool.}
The NE of Scenarios 1 and 2 show that, in principle, money can be used at a microeconomic 
level to coordinate network packets. By definition, the flows of the scenarios can only benefit through the use of money; each trade is a weak Pareto improvement for the current state of the economy.
However, the benefit for the individual flows and the overall network has to
be further studied. We present some preliminary results on the underlying single 
machine scheduling problem in Section~\ref{sec:sched} and the effect of trades in Section~\ref{sec:effect}.

\noindent
{\bf Packet Pairing.}
A core operation of the PacketEconomy is the random pairing of the packets 
that takes place in each trading period to generate the trading pairs.
We present an efficient parallel algorithm that can support the random pairing 
procedure in real time. 

The new algorithm is a parallel, or better, a pipelined version of 
the random shuffling algorithm of 
Fisher-Yates, which is also known as Knuth shuffling~\cite{wiki:2011:shuffling,Knuth:1981:Vol2}.
The Fisher-Yates shuffling technique was introduced in~\cite{Fisher:1948:Shuffle},
later Durstenfeld~\cite{Durstenfeld:1964:Shuffle} proposed a corresponding
$O(n)$ algorithm, and finally Knuth~\cite{Knuth:1981:Vol2} popularized
Durstenfeld's algorithm.
The random shuffling of Fisher-Yates is a simple and elegant way to generate
a random shuffle with a single pass over an array of items.
We call the new algorithm \emph{Pipelined Shuffling} (see Algorithm 1). Its core
is a pipeline of $q$ instances $0,1,\dots,q-1$ of the Fisher-Yates algorithm 
At time $t$, instance $k$ is at step $t+k \bmod q$ of the random shuffling algorithm.
If there is a processing unit for each queue position then the parallel step
can be executed in $O(1)$ time. The pipeline delivers a random permutation
in each trading period. The random permutation can be used to obtain a 
random pairing in $O(1)$ parallel time with $q$ processors.

\begin{localscope}
\begin{algorithm}
\caption{Shuffling}
\label{alg:shuffle}
\begin{algorithmic}[1]
\Procedure {Shuffle}{int[] a}
\cw{2cm}

len = a.length;

for i from 0 to a-1 do \{ 

\hspace{0,5cm} \textit{// i==a-1 is a dummy loop}

\hspace{0,5cm}      j = random int in $i \leq j \leq q-1$; 

\hspace{0,5cm}       exchange a[j] and a[i]

\} 

\EndProcedure

\end{algorithmic}

\begin{algorithmic}[1]
\Procedure {PipelinedShuffling}{int[][] A}
\cw{2cm}

len = A.length;

for i from 0 to len-1 do in parallel \{ 

\hspace{0,5cm}      processor i: wait for i periods;

\hspace{0,5cm}      processor i: while (true) \{Shuffle(A[i]);\} 

\}
\EndProcedure
\end{algorithmic}
\end{algorithm}
\end{localscope}

\begin{lemma}
\label{lem:shuffle}
The Pipelined Shuffling algorithm runs in $O(1)$ parallel time on a $q$ processors EREW PRAM
and delivers a random shuffle every $O(1)$ steps.
\end{lemma}

\begin{proof}
\hla{A running instance of the Pipelined Shuffling algorithm contains $q$ independent instances
of the basic Shuffle algorithm. 
Each Shuffle instance is executed by one of the $q$ processors.
From the pseudocode of the algorithms we can conclude that
each instance of the Shuffle algorithm is at a different round of its main loop.
Moreover, each instance of the Shuffle 
algorithm has its own vector of $q$ memory positions
to store its current permutation and, thus, there is no possibility of two processors 
concurrently accessing the same memory position. In each round, one Shuffle instance completes its execution and delivers a random permutation of the $q$ numbers $\{1, 2, \dots, q \}$.}
\end{proof}

\noindent
The PacketEconomy packet pairing algorithm uses the delivered random permutation 
to generate a random pairing in $O(1)$ parallel time on $\lceil q/2 \rceil$ processors.

\begin{theorem}
\label{the:pairing}
A random packet pairing can be generated every $O(1)$ parallel time on a $q$ processors EREW PRAM.
\end{theorem}

\begin{proof}
\hla{
From Lemma~\ref{lem:shuffle} we know that a random permutation can delivered 
with Pipelined Shuffling every $O(1)$ parallel steps.
The algorithm to generate a random pairing from it
requires $\lceil q/2 \rceil$ parallel processors and works as follows.
A separate vector of $q$ memory positions is used to store the final pairing.
Each processor $i$ of the involved processors reads the values $x_{2i}$ and $x_{2i+1}$
of the positions $2i$ and $2i+1$ of the permutation, respectively, and then writes into position 
$2i$ of the pairing vector the value $x_{2i+1}$ and into position 
$2i+1$ the value $x_{2i}$. If $q$ is an odd number, then one position will not be paired. 
The contents of the final vector specify for each position the corresponding paired position.}
\end{proof}

\noindent
{\bf Computational Requirements of the PacketEconomy.}
A very important advantage of the PacketEconomy is that the computational
requirements of the flows and the router are tolerable.
In contrast to solutions that implement global auctions between the flows,
the PacketEconomy is based on simple interactions between paired packets.
Assuming that the price for each packet trade can be calculated in $O(1)$ time, 
in each trading period the work is $O(1)$ for each packet 
that belongs to a packet pair, in total $O(q)$.
However, all operations are local to the pairs of queue positions.
With appropriate, fairly simple multi-core hardware, each round can be 
executed in $O(1)$ parallel time. Further arguments which support the 
applicability of our approach to modern routers are presented in 
Section~\ref{sec:hw}.

\section{The Scheduling Problem}
\label{sec:sched}

The underlying algorithmic problem of the PacketEconomy is a scheduling problem of network packets.
From the router's point of view, this problem is a single machine scheduling problem with a max 
weighted total wealth objective.

\begin{definition} Max-Total-Wealth Scheduling (MTW).
$N$ jobs $j_i$, for $i=1,\dots,n$. Job $j_i$ has processing time $p_i$, release date $r_i$,
deadline $d_i$ and weight $w_i$. Let $c_i$ be the completion time of job $i$ in a schedule.
The objective is to find a non-preemptive schedule that maximizes the total wealth $W = \sum_i w_i \cdot \max(d_i - c_i, 0)$.
\end{definition}

\noindent
The release date $r_i$ is the time when packet $i$ enters the queue
and the deadline $d_i$ is the time when the value of the packet becomes zero.
Note that the price used for each trade does not directly influence the sum $W$ (the total wealth).
For MTW on a network router the following assumptions hold:
\begin{enumerate}
\item The queue discipline is work-preserving, meaning the router is never left idle if the queue is not empty.
\item The number of packets that can be in the queue at any time is bounded by a constant (the maximum queue size). 
\item The packet sizes may differ by at most a constant factor. In this work we assume that all packets are of the same size.
\end{enumerate}

\noindent
The complexity of the scheduling problem strongly depends on the assumptions made. 
Without deadlines, i.e., without a limit on the delay of each packet, 
an optimal schedule can be obtained by applying a greedy rule like Smith's rule~\cite{Smith:1956:Rule}. In particular, the router may simply serve in each 
round the packet with the largest cost factor $c_i$. 
\begin{theorem}
The MTW problem without deadlines can be optimally solved in polynomial time.
\end{theorem}

\noindent
This holds even for the online version of the problem where the router knows only
the packets in its queue; the greedy rule gives a $1$-competitive algorithm.
\begin{theorem}
There is $1$-competitive algorithm for MTW without deadlines.
\end{theorem}

\noindent
However, in realistic scenarios with IP packets, there are deadlines. Common IP packets 
have a time-to-live (TTL) field. In TCP, a packet that is not acknowledged within the specified
timeout period is considered lost. The scheduling problem for packets with deadlines
can be solved off-line as a linear assignment problem (LAP), where packets are assigned to time-slots 
(rounds).
This approach is used in~\cite{GLLR79} for a min-weighted-tardiness problem 
that is related to the MTW problem.

\begin{theorem}
The MTW problem with deadlines can be optimally solved in polynomial time.
\end{theorem}

\subsection{Bandwidth Sharing}
Due to the on-line nature and the finite queue size of the PacketEconomy router, 
the above classic scheduling algorithms do not seem to naturally fit the 
MTW scheduling problem of the PacketEconomy queue. An additional reason is that
for window-based flows, packet transmission is a closed loop. 
Consequently, the order in which the queued packets are served influences, if not determines,
the next packet that will enter the queue. Thus, even the online assumption may not be appropriate.
\hla{A different approach to study the scheduling problem of the PacketEconomy is to 
focus on how the bandwidth is shared between the packets.} In this case we consider the 
(average) packet rate of the flows, as shown in the following example.

\begin{example}
\label{exa:rates}
Assume a scenario with window-based flows and 5 economy packets and 5 business packets.
There is a deadline of 40 rounds on the maximum delay of the economy packets.
Moreover, the business packets have to be treated equally. The same holds for 
the economy packets. Consider the scenario where each economy packet will be served
with a rate of 1/40 packets/round and delay of 40 rounds and the business flows share 
the remaining bandwidth; each business packet is served at a rate of 7/40 packets/round
and delay 40/7 rounds. This is an upper bound on the rate of total wealth for the router
for this scenario. 
\end{example}

\section{The Effect of Trades}
\label{sec:effect}

In the PacketEconomy, each packet can increase its utility by making trades. 
To show the potential
of the approach, consider a packet of maximum priority that pays enough to 
make any trade that reduces its delay.
In the analysis we will assume that the probability of packet failures is very low, 
and thus ignore it. We will focus on window-based flows.
We first present an exact calculation for the average delay of this packet 
and then derive simpler, approximate bounds on the delay.

\begin{lemma}
\label{lem:delayAVG}
The average delay $E[d]$ of the packet is 
\begin{equation}
E[d] = \sum_{s=0}^q s \left(\frac{1}{q-2}\right)^s (s-1) \frac{(q-2)!}{(q-s-1)!} \; . 
\end{equation}
\end{lemma}

\begin{proof}
Let $rand(L, U, s)$ be a uniformly random integer in $\{L, L+1, ..., U\} \backslash \{s\}$
and $pos(p)$ the current position of packet p. Then, the probability $Pr[d > s]$ is
\[ Pr[d > s] = \prod_{k=1}^s Pr[rand(1,q-1, pos(p)) \geq s-k+1]  \Rightarrow \]
\[ Pr[d > s] = \frac{q-s-1}{q-2} \cdot \frac{q-s}{q-2} \cdots \frac{q-2}{q-2} \Rightarrow \]
\[ Pr[d > s] = \left(\frac{1}{q-2}\right)^{s} \cdot \frac{(q-2)!}{(q-s-2)!} \; . \]

\noindent
We can now calculate the probability density function $Pr[d=s]$ of the delay $s$.

\[ Pr[d = s] = Pr[delay \leq s] - Pr[delay \leq s-1] \Rightarrow \]

\[Pr[d = s] = \left(\frac{1}{q-2}\right)^{s-1} \cdot \frac{(q-2)!}{(q-s-1)!} - \left(\frac{1}{q-2}\right)^s \cdot \frac{(q-2)!}{(q-s-1)!}  \Rightarrow \]

\[Pr[d = s] = \left(\frac{1}{q-2}\right)^s \left(s-1\right) \frac{(q-2)!}{(q-s-1)!} . \] 

\noindent
Applying the definition of the expected value completes the proof
\[E[d] = \sum_{s=1}^{q} s \left(\frac{1}{q-2}\right)^s \left(s-1\right) \frac{(q-2)!}{(q-s-1)!} . \]
\end{proof}

\begin{lemma}
\label{lem:delayUB}
The average delay of the packet does not exceed $\frac{-1+2 b+2 \sqrt{2b (q-2)}}{2 b}$. For $b=1$ the bound is $\frac{1}{2} + \sqrt{2(q-2)}$.
\end{lemma}
\begin{proof}
A packet that enters at position $q-1$ has been served when it advances 
at least $q$ positions. Assume $b=1$ trading periods per round.
Each round starts with a shift of all packets one position ahead followed by $b$ trading periods.
Note that each random trading partner corresponds to a uniform random number
in $[1,q-1]$. To admit a more elegant mathematical treatment we prefer the continuous 
distribution. The absolute difference between the expected values of the discrete and the continuous 
variables is not larger than one. The same bound holds for the difference between the 
expected values of the minimums after $k$ draws.

Assume a random variable $X_c$ that is uniformly distributed in $[L, U]$,
where $L$ and $U$ are integers, such that $L < U$. Let $X_d$ be a random variable
that is calculated from $X_c$ in the following way:
\[
X_d = L + i, \mathrm{ where } \; i \; \mathrm{is\;such\;that}: L + i \cdot A \leq X_c \leq L + (i+1) \cdot A \; ,
\]
where $A = \frac{U-L}{U-L+1}$. 
The random variable $X_d$ corresponds to a uniform discrete random variable in $\{L, L + 1, \dots, U\}$. The absolute difference $X_c - X_d$ is not larger than one.
Consequently, the absolute difference between the minimum $X_{\min}^c$ of $m$ draws of $X_c$ and the corresponding minimum $X_{\min}^d$ of the $m$ values of $X_d$ is also not larger than one. 
The same bound holds for the difference between the expected values of the minimums after $k$ draws.
Thus, we obtain that
\begin{equation}
\label{equ:AvgMaxDiscContOne}
E[X_{\min}^c] - 1 \leq E[X_{\min}^d] \leq E[X_{\min}^c] + 1 \; .
\end{equation}

\noindent
Moreover, note that $k >= 1$, the average minimum of $k$ random draws will be in the lower
half of the interval $[L,U]$, that is in $[L,\frac{L+U}{2}]$. Real values in this interval 
are on average rounded to smaller integer values in the above rounding procedure for $X_c$
to $X_d$. Thus, the average discrete minimum will not be larger than the average
continuous minimum. Thus,
\begin{equation}
\label{equ:AvgMaxDiscCont}
E[X_{\min}^d] \leq E[X_{\min}^c] \leq E[X_{\min}^d] + 1 \; .
\end{equation}

\noindent
As shown in the following lemma, the average value of the minimum of these $k$ draws is $(q-2)/(k+1)$.
\begin{lemma}
\label{lem:min}
Let $X_1, X_2, \dots, X_k$ be continuous uniform random variables in $[0,U]$. 
Let $X_{\min}$ be the minimum of these variables. Then $E[X_{\min}] = \frac{1}{k+1} U$.
\end{lemma}

\begin{proof}
The probability distribution of each continuous uniform random variable $X_i$ is $F_{X_i}(x) = \frac{x}{U}$.
The probability distribution of the minimum $X_{\min}$ is 
\[F_{X_{\min}}(x) = 1 - \prod_{i=1}^k \left( 1-F_{X_i}(x)\right) \; . \] 
Now, applying the definition of the expected value function completes the proof.
\end{proof}

Back to the proof of Lemma~\ref{lem:delayUB}, we assume that a packet has just entered the 
queue at position $q-1$.
Consider $k$, such that the average delay of the packet is $k+1$ rounds. 
Then, after the $k$ rounds and $b k$ draws the packet advances for $h$ additional rounds
until it is served. Then the total delay of the packet is $k+h+1$. 

The average number of draws until it achieves its minimum is $k/2$.
We add one to the average minimum, because the minimum position that can be traded 
is position $1$. Position $0$ is the one that is currently served.
Let $b$ be the number of trading periods per router round. 

\noindent
Thus, we have

\[ \frac{1}{b k + 1} (q-2) + 1 - k/2 - h - 1 \leq 0 \; .\]

\noindent
We solve for $k$ and obtain that the larger of the two roots of k is 

\[k = \frac{-1-2 b h+\sqrt{1-16 b-4 b h+4 b^2 h^2+8 b q}}{2 b} \; . \]

\noindent
The total delay $k+h+1 = \frac{-1+2 b h}{\sqrt{1+4 b^2 h^2-4 b (4+h-2 q)}}$
is minimized at $h = \frac{1}{2b}$ where
\[
k+h+1 = \frac{-1+2 b+2 \sqrt{2} \sqrt{b (-2+q)}}{2 b} \; .
\]

\noindent
For $b=1$ we get that the minimum value of $k+h+1$ is 
\[
\frac{1}{2}+\sqrt{2(q-2)} \; . 
\]

\noindent
The average delay $k+1$ cannot be larger then the above value. This completes the proof of Lemma~\ref{lem:delayUB}.
\end{proof}

The above lemma can be generalized to the case where only one packet in every $c > 0$ packets in the queue is ready to sell its position. We simply assume $b/c$ trading periods per round. 
In this case the average delay of the business packet is not larger 
than $1-\frac{c}{2}+\sqrt{2 c (q-2)}$.

\begin{lemma}
\label{lem:delayLB}
The average delay of the packet is at least $\frac{-1+2b+\sqrt{1-8 b+4 b q}}{2 b}$.
For $b=1$ the bound is $(1/2) +\sqrt{4q - 7}$.\footnote{\hlb{There was a minor error 
in the expression of the bound of this lemma in a previous version of this work. More precisely,
the bound for $b=1$ was $1 +\sqrt{q - 1}$ and the minor error was that the constant $1$ caused by 
using the continuous minimum instead of the discrete minimum was not subtracted from the left hand side 
of Equation}~\ref{equ:delayLBfirst}.}
\end{lemma}

\begin{proof}
Assume $k+1$ rounds with $b=1$. 
The continuous average minimum of $k$ rounds with $b$ trading periods is $(q-2)/(bk+1)$.
From Equation~\ref{equ:AvgMaxDiscCont} we obtain that the average discrete minimum is at 
least $(q-2)/(bk+1) - 1$. 
We will add one to this number because the minimum possible trade is position 1.
In the best case the minimum is achieved with the first draw. In the remaining $k$ rounds the packet will be (in any case) shifted by $k-1$ positions (in each round except the round when it entered the queue).
This number of simple steps/shifts is subtracted from the min. 
Finally, a delay of one round is needed to serve the packet, when it arrives at position 0.
Consequently,

\begin{equation}
\label{equ:delayLBfirst}
\left(\frac{1}{bk+1} (q-2) - 1 \right) + 1 - (k-1) - 1 \leq 0 \; .
\end{equation}

\noindent
From the above inequality and the fact that $k$ is positive we obtain 

\[
k \geq \frac{-1+\sqrt{1-8 b+4 b q}}{2 b} \; .
\]

\noindent
Using $b=1$ the expression is simplified to 
$k \geq -(1/2) + \sqrt{4q-7}$. Thus the average delay is at least

\[
k + 1 \geq \frac{1}{2} + \sqrt{4q-7} \; .
\]
\end{proof}

The previous lemma can also be generalized to the case where only one packet in every $c > 0$ packets in the queue is ready to sell its position. In this case the average delay of the business packet is not less than $\frac{1}{2}(2-c)+\frac{1}{2} \sqrt{c^2+ - 8 c + 4 q c}$.

In Figure~\ref{fig:Delays}, analytical and experimental results for the delay of the business packet are presented.

\begin{figure}[!h]
\subfloat[The exact average delay (inner line) and the lower and upper bounds on the average delay]{
\label{fig:delaysXXX}\includegraphics[width=0.4\textwidth]{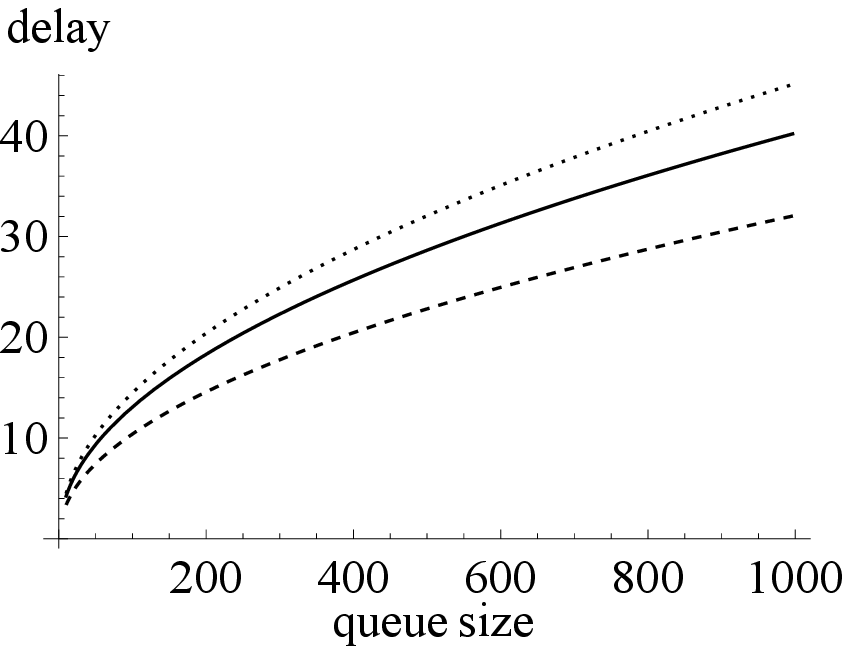}}
\hspace{0.1\textwidth}
\subfloat[Experimental measurement of the delay.]
{\label{fig:valuesXXX}\includegraphics[width=0.44\textwidth]{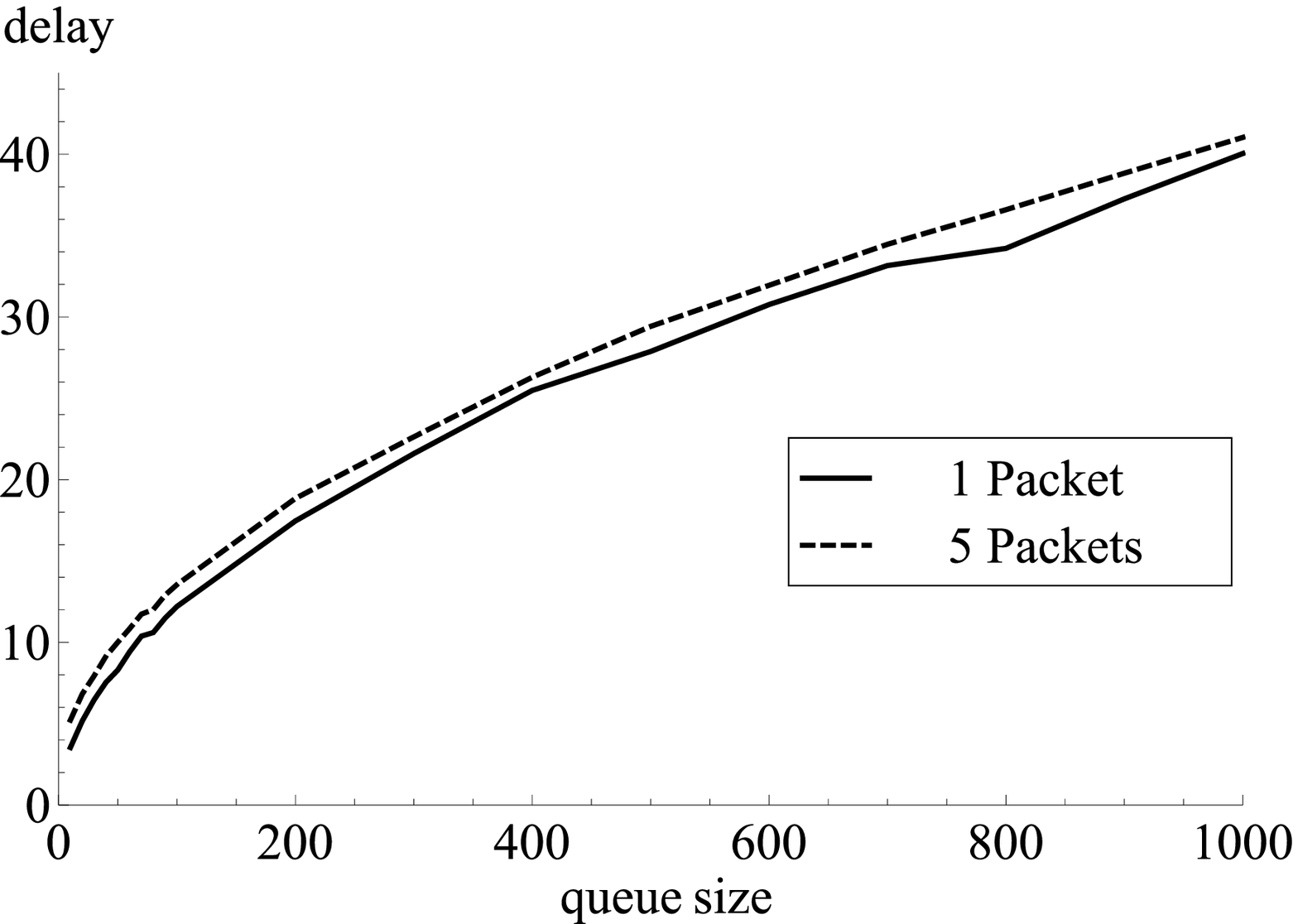}}
\caption{Delay of the business packet with respect to the queue size} \label{fig:Delays}
\end{figure}

\subsection{An Example of the Effect of Trades}
\label{sec:anarchy}

The effect of bilateral trades on the overall convergence of an economy is studied for 
example in~\cite{Feldman:1973:Bilateral}. In this section we examine the social wealth 
of the PacketEconomy for the cases of no trades, trades and an ideal scheduling of the packets.  
In our model we do not expect to obtain an optimal solution for every queue snapshot. 
This would be too ambitious
given the real-time work load of network routers and the requirement 
for $O(1)$ (parallel) time of processing per round. 
However, we hope to achieve significant improvements with a lightweight 
$O(1)$ time procedure. The theoretical arguments and the experimental
evaluation show that PacketEconomy can support better coordination between selfish network packets 
which leads to significant improvements for both flows and routers.

We assume window-based flows and two types of packets, economy packets and business packets.
We examine and compare the cases of no-trades, trades and ideal trades.
The social wealth or wealth rate is the total utility per round. Note that the money 
circulating in the economy remains constant and thus the total utility depends on the 
values of the packets.

A simple but somewhat surprising result is that if the max value
of economy packets $V_e$ and the max value of business packets $V_b$ are
equal then the social wealth is independent from the scheduling algorithm.
Any schedule that respects the packet deadlines achieves the same
social wealth.

\begin{lemma}
\label{lem:SocialWealthEqualVmax}
Assume a router queue with two-types of window-based flows, economy flows and business flows. 
If $V_e = V_b$, then the social wealth is the same for any feasible schedule (by feasible we mean 
a schedule where no packet exceeds its delay limit).
\end{lemma}

\begin{proof}
The proof is a direct application of Equation~\ref{equ:priceWindowBased}.
\end{proof}

Let $n_e$ be the number of economy packets, $V_e$ their max value, $c_e$ their the cost per round of delay, and $d_e$ their delay. Similarly, we have $n_b$, $V_b$, $c_b$, and $d_b$ for the business packets.
The size of the queue is $q = n_e + n_b$. 

\subsubsection{No Trades}
When no trades take place then all packets experience the same delay $d_e=d_b=d$, which is equal to the
queue size $q$. The rate for each packet is $1/d$.
The social wealth is utility of the economy packets plus the utility of the business packets.
Each economy packet generates a benefit of $V_e-c_e q$ in each round-trip (from its submission until it is served). The round-trip time of the packet is equal to its total delay $d_e$, which is equal to $q$. The same holds for the business packets. 

Thus, the delays are $d_b = d_e = q$ and the social wealth (total utility generated) is 
\begin{equation}
\frac{n_e}{q}(V_e-c_e q) + \frac{n_b}{q}(V_b-c_b q) \; .
\end{equation}

\subsubsection{Ideal Trades}
In the ideal case, the economy packets will consume all their delay and release in this way 
as much bandwidth as possible. Note that even two business packets could consume up to the
whole bandwidth (if the other packets sell).
The available bandwidth is then evenly distributed among the 
business flows.
The social wealth is 
\begin{equation}
\frac{n_e}{4q}(V_e-c_e 4q) + \frac{4q-n_e}{4q}(V_b-c_b \frac{n_b 4 q}{4q - n_e})
\end{equation}

\subsubsection{PacketEconomy Trades}
\label{sec:PETrades}
Let us start with the case of a queue with one business packet, that is, $n_b = 1$ and $n_e = 99$.
From Lemma~\ref{lem:delayUB}, we obtain that $d_b \leq 1/2 + \sqrt{2(q-2)} = 14.5$.
The computationally heavy exact computation of the expected value gives $E[d_b] = 13.08$ and experimental results give $\bar{d_b}=12.8, \sigma(d_b)=6.7$.

Let us now consider the case of a queue with $n_b = 5$ business packets and $n_e = 95$ economy packets.
From Lemma~\ref{lem:delayLB} we get $d_b \geq 10.41$, since the 
the delay of each business packet cannot be smaller than for the case $n_b = 1$. Experimentally, $\bar{d_b}=14.5, \sigma(d_b)=7.4$. Then

\[
r_b = \frac{n_b}{d_b} \leq \frac{5}{10.41} \leq 0.48 \quad \mathrm{and} \quad
 r_e = 1 - r_b \geq 0.52.
\]
Consequently, the probability that the packet at position $i$ is an economy packet is at least $0.54$.
Moreover, the average delay of an economy packet is 
\[
d_e = \frac{n_e}{r_e} \leq \frac{95}{0.52} \leq 182.7 < 400. 
\]
Experimentally, $\bar{d_e}=144.7, \sigma(d_e)=61.5$. That is, on average an economy packet is not expected to sell more than $83$ rounds of delay.
Thus, we can assume that almost all economy packets will not exhaust their extra delay of $300$ rounds.
Consequently, every time the business packet is paired with a preceding economy packet,
the economy packet will be able to trade its position.

However, what is the probability that the trading partner of a business packet 
will be an economy packet? There are constantly $n_e = 95$ economy packets in
the queue, but in general these packets will not be distributed uniformly within
the $q=100$ queue positions. For example, the concentration of economy packets might
be higher near the end of the queue.

Assume an arbitrary queue round. 
For $i=0,1, \dots, 99$, let $p_{e,i}$ be the probability that the packet in queue position $i$
is an economy packet. First note that $p_{e,0} = r_e \geq 0.52 \;. $
Moreover, every economy packet that reaches position $0$, must have passed through every other 
queue position at least once. This means that for each queue position, 
the rate of economy packets must be at least $r_e$,
else the global rate $r_e$ would not be sustainable. Thus
\[ \mathrm{for} \; \mathrm{all} \; i=0,1, \dots, 99, \; p_{e,i} \geq 0.52 \; . \]
We can apply Lemma~\ref{lem:delayUB} to upper bound the delay $d_b$ 
of the business packets. This time we use $c = 1/0.52$ and obtain
\[d_b \leq 19.45 \;. \]
Again, the delay is significantly reduced compared to the delay of 100 in the case of no trades.
In such scenaria with two types of packets, the business packets may reduce 
their delay from $q$ to $O(\sqrt{q})$, which is a significant improvement.
It is of course an interesting problem to study the behavior 
of the PacketEconomy approach in more complex scenaria
where the queue may contain a blend of several packet types
and examine if then improvements of $\sqrt{q}$ are still possible.

Finally, substituting the indicative values $V_e = 500, V_b = 1000, c_e = 1, c_b = 4$ 
and using the bandwidth sharing approach shown in Example~\ref{exa:rates}
we can calculate the wealth rate of the economy in the example.
Recall that the size of the queue is $q=100$ and that the maximum delay for
any packet is $4q = 400$.

\hla{First we consider the case of $n_b=1$ business packet and $n_e=99$ economy packets.
In the case of no trades, i.e., a simple FIFO queue, all packets experience a delay 
of $100$ rounds. 
In this case $r_b = 0.01$ and the value of the business packet is $600$.
The total rate for the economy packets is $r_e = 0.99$ and the value of each 
economy packet $400$. Combining the above gives that the average wealth per 
round generated by the economy is $402$ (also $402$ experimentally). 

In the ideal case all economy packets will consume their delay of $400$ rounds
to free bandwidth for the business packet. This would mean
$r_e = n_e/d_e = 0.2475$. Consequently $r_b = 1 - r_e = 0.7525$
and $d_b = n_b/d_b = 1.33$. The corresponding wealth rate for the economy 
would $773.25$. Recall that the transfer of budget from one packet to another 
does not change the total wealth, and, consequently,
we did not consider how budget circulates.

In the above ideal case the delay for the business packet is $1.33$
which is practically not possible for the PacketEconomy. For $q \geq 2$, the delay 
of any packet is at least $2$; one round to enter the queue and one round to be served. 
The best feasible delay for the business packet would be $d_b = 2$.
Doing the calculations gives a wealth rate of $647$ per router round. 
In experiments with $1000$ trading periods, i.e., $10$ times the size of the queue, 
we achieved a wealth rate of $644.6$ per router round.

For the case of the PacketEconomy with trades with random pairing and 
one trading period per round we use the bounds on the packet rates and
delays obtained earlier. This gives that the wealth rate is $\geq 431.52$
per round (experimentally, $436.1$).

Similarly, for $n_b=5$ we obtain the wealth rates $410$ (also $410$ experimentally), $766.25$ ($766.2$ experimentally with $1000$ trading periods) and $\geq 515.75$ ($556.6$ experimentally), 
for no trades, ideal trades, and PacketEconomy, respectively.
Note that the important feature are the self-adjusted packet delay times of 
the economy and not the wealth rate itself.
}

\section{Experimental Evaluation}
\label{sec:exp}

We have performed \hla{two} sets of experiments in order to verify and quantify our theoretical predictions. The first set of experiments simulates the scenario presented in Section~\ref{sec:ScenarioNoFiat}.

\begin{figure}[h!]
\centering
\includegraphics[width=\textwidth]{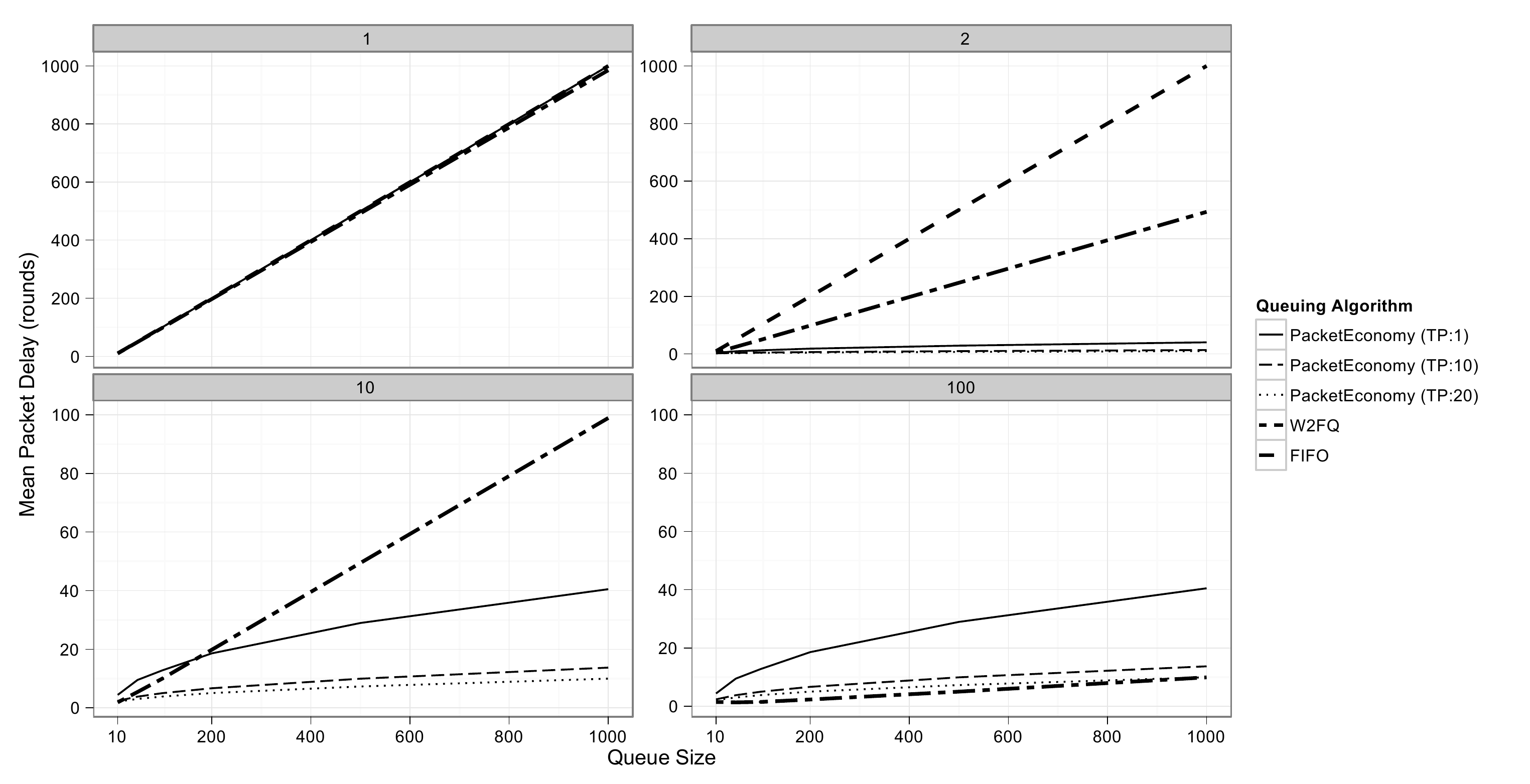}
\caption{Mean business flow packet delay ($d_b$) comparing PacketEconomy (trading periods $b\in\{1,10,20\}$) vs FIFO and $W^{2}FQ$. The figure consists of four panels for four values of the business flow cost per round of delay $c_b\in\{1,2,10,100\}$. In the two bottom panels ($c_b\in\{10,100\}$) the FIFO flow has been removed for enhanced clarity.}
\label{fig:ExpDelay}
\end{figure}


In this experiment we compare the delay of one one-packet window business flow in increasingly large router queues ($q\in\{10,..,1000\}$). Firstly, we compare the mean packet delay of this flow when using PacketEconomy with $c_b\in\{1,..,100\}$  to using the $W^{2}FQ$ (Worst-case Fair Weighted Fair Queueing) ~\cite{BZW2FQ} algorithm with weights for the flow equal to $c_b$  and to using a simple FIFO queuing policy. In the experiments $c_e=1, V_e=4q, V_b=4qc_b$ and for each run we simulated serving $100q$ packets. The results shown in Figure \ref{fig:ExpDelay} show a subset of values for $c_b$ since FIFO always performs the same. Also, the PacketEconomy algorithm also performs the same for all $c_b>1$ and is only different for $c_b=c_e=1$ since then the delay is identical to FIFO. Finally, $W^{2}FQ$ progressively decreases the business flow's mean delay as its weight ($=c_b$) increases. The conclusion which can be drawn from these experimental results is that PE outperforms FIFO and that it is comparable to $W^{2}FQ$ for \hla{$b>c_p$}.

\begin{figure}[h!]
\centering
\includegraphics[width=\textwidth]{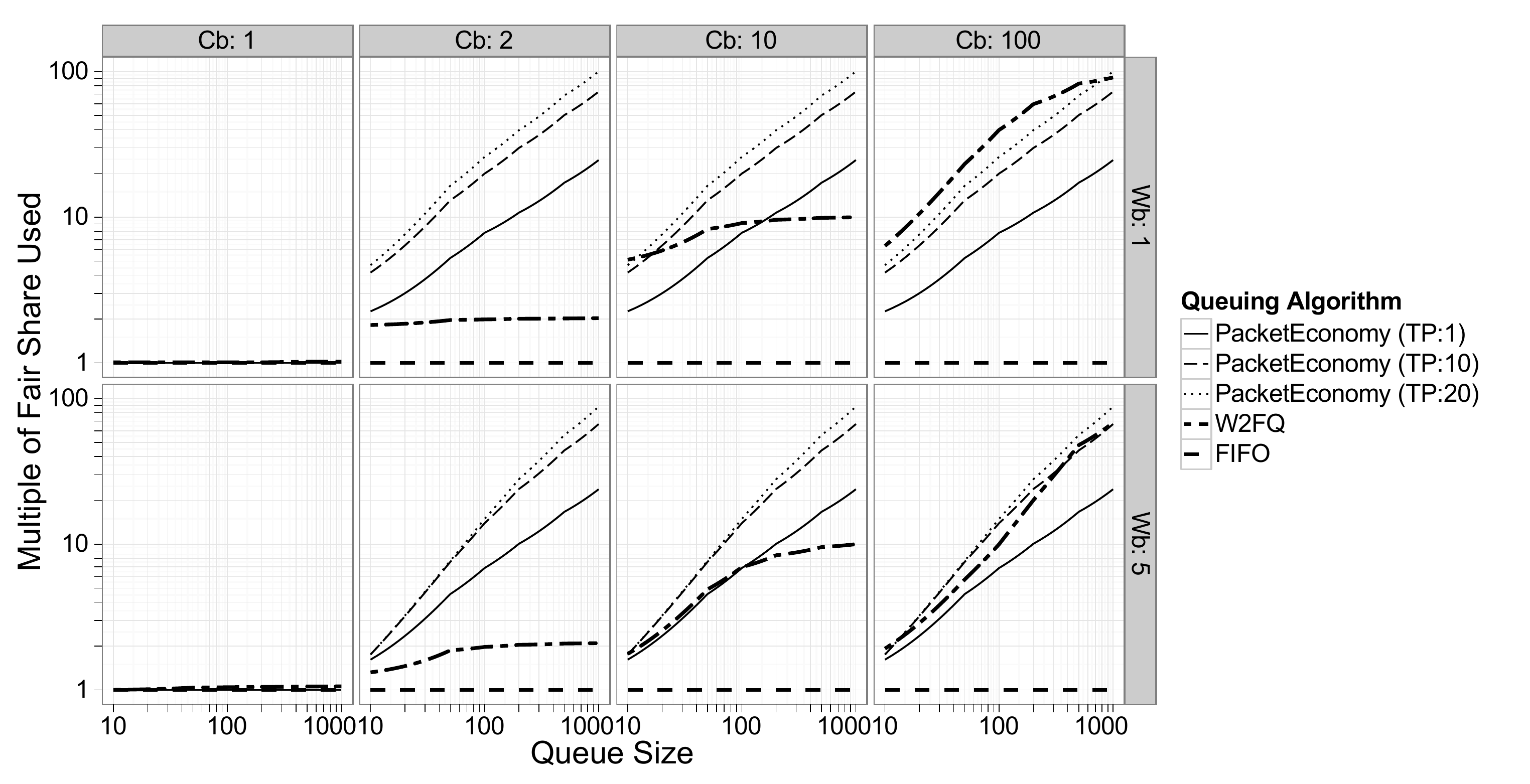}
\caption{Multiple/ratio of Fair Share used by the business flow comparing PacketEconomy (trading periods $b\in\{1,10,20\}$) vs FIFO and $W^{2}FQ$. The plots are displayed log-log to highlight power functions. The columns vary the business flow cost per round of delay $c_b$ value and the rows the business flow window size $n_b$ value.}
\label{fig:ExpFairShare}
\end{figure}

Secondly, we examine how much traffic is consumed (packets sent) by the business flow by comparing the flow's share of the packets sent to the flow's fair share. For the same set of parameters described above, plus the case for $n_b=5$, we determine the fair share of a flow $f_i$ to be $h_i=w_i/q$ and the we calculate the ratio of the fair share taken as
\[ 
\lambda_{h_i}=\frac{a_i}{\sum_{i}^{}{a_i}}\frac{1}{h_i}
\]
where $a_i$ is number of packets delivered by flow $f_i$. This metric expresses how much more traffic than it would normally (with FIFO) be able to take up the flow managed to acquire. 
The results shown in Figure \ref{fig:ExpFairShare} indicate that, rather obviously, for $c_b=1$ the business flow gets its fair share under all algorithms. However, as $c_b$ increases, PacketEconomy is able to provide exponential ratios of the fair share and for \hla{$b>c_b$} can even be comparable to $W^{2}FQ$.


The \hla{second} set of experiments examines the quality of random pairing and trading as a scheduling algorithm. We used a simple queue with business and economy packets. The business packet utility costs four times more per round of delay compared with the economy packets. Three cases were examined: a) no scheduling, i.e., packets are served in order, b) optimal scheduling and c) random pairing scheduling. Optimal scheduling finds the optimal solution to this problem. 

\begin{figure}[h!]
\begin{minipage}[b]{\linewidth}
\centering
\includegraphics[width=\textwidth,height=0.4\textwidth]{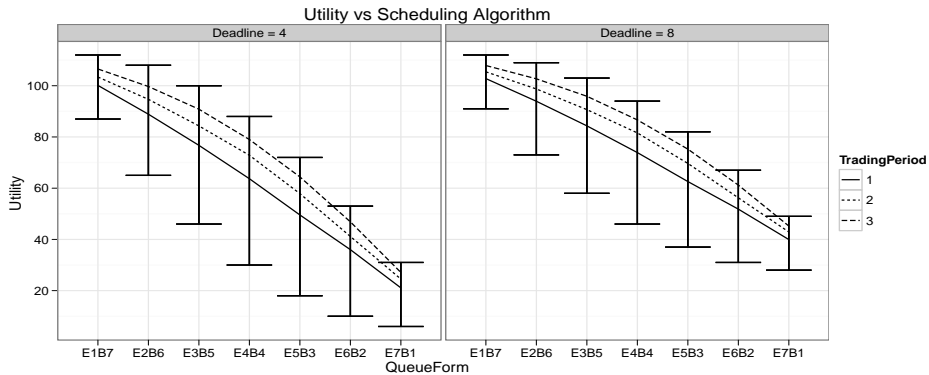}
\caption{Utility of random pairing scheduling in comparison with optimal scheduling(max) and no scheduling(min). $ExBy$ denotes $x$ economy packets followed by $y$ business packets}
\label{fig:Scheduling}
\end{minipage}
\end{figure}

For the random pairing scheduling case, we performed one, two or three random trading periods per round to investigate the extent to which multiple trading periods will improve performance. Two sub-cases were examined, where the deadline of the economy packets is 4 or 8 rounds. The deadline for the business packets was always 8 rounds. We also examined whether different queue compositions would give different results and therefore tested all the possible business and economy packet queue compositions.

In Figure \ref{fig:Scheduling}, the error bars represent the optimal scheduling utility (max) and the no scheduling (min) utility. It can be seen that in all the cases random pairing performs close to the optimal solution. Additionally, increasing the number of trading periods yields improved results.

\section{Conclusion}
\label{sec:conclusion}
We presented a network economy and showed the existence of NE 
where money circulates to the benefit of the flows.
The computational requirements of a router that would implement
the PacketEconomy approach are at an acceptable level. This is very important
since network routers have to process massive streams of network packets in
real time.

There are several other issues that have to be addressed
for our model to be of practical importance. For example,
a greedy flow may submit economy packets to the network simply to collect
money. A realistic economic model has to anticipate such scenarios and address 
them with appropriate rules. For example, a general rule could be that the final 
final budget of any packet could be restricted to be non-positive. A more effective 
rule could impose router-entry costs on every packet that depend on the current load 
of the router. \hla{The general topic of countering malicious behavior is of-course 
a never ending game between service providers and legitimate users on the one side 
and malicious entities or users on the other side.
Note that even simple and established flow control algorithms like AIMD 
are prone to malicious flow behavior.}
Finally, the burden of accounting has also to be handled.
For example, determining which network entities will be responsible for managing currency 
issues.

Overall, we examined how money can be used at a microeconomic level 
as a coordination tool between network packets and we believe that 
our results show that the PacketEconomy approach defines an 
interesting direction of research for network games. We 
are currently working on implementing the PacketEconomy in a 
realistic network context.

\vspace{0.2cm}
\noindent
{\bf Acknowledgements.}
The first author, is grateful to Paul Spirakis for inspiring discussions on
the intersection of Algorithms and Game Theory. We also wish to thank Vasilis Tsaoussidis
for insightful discussions on using budgets in Internet router queues.

\end{document}